\documentclass[aps,twocolumn,showpacs,pra,nofootinbib]{revtex4}
\usepackage{bm}
\usepackage{graphicx}
\usepackage{amsbsy}
\usepackage{amsmath}
\usepackage{amsfonts}
\usepackage{amssymb}
\usepackage{amsmath}
\usepackage{amsthm}
\usepackage{url,doi}
\usepackage{hyperref}
\usepackage{xspace}

\newtheorem*{lemA.1}{Lemma A.1}
\newtheorem*{lemA.2}{Lemma A.2}
\newtheorem*{lemA.3}{Lemma A.3}
\newtheorem{theorem}{Theorem}

\newtheorem{lemma}{Lemma}
\newtheorem{definition}{Definition}

\newtheorem{problem}{Problem}

\newcommand{\color}[3]{}

\newcommand{\op}[1]{\mathsf{#1}}

\begin{document}
\title{Gaussian quantum computation with oracle-decision problems}
    \author{Mark R. A. Adcock,$^{1}$ Peter H\o yer,$^{1,2}$ and Barry C. Sanders$^{1}$
   }
 \affiliation{$^{1}$Institute for Quantum Information Science,
    University of Calgary, Calgary, Alberta, Canada, T2N 1N4. Email: \texttt{mkadcock@qis.ucalgary.ca}\\
 $^{2}$Department of Computer Science,
University of Calgary, 2500 University Drive N.W., Calgary,
Alberta, Canada, T2N 1N4. Email: \texttt{hoyer@ucalgary.ca} }
%======================================================================
\begin{abstract}
We study a simple-harmonic-oscillator quantum computer solving
oracle decision problems. We show that such computers can perform
better by using nonorthogonal Gaussian wave functions rather than
orthogonal top-hat wave functions as input to the information
encoding process. Using the Deutsch–-Jozsa problem as an example,
we demonstrate that Gaussian modulation with optimized width
parameter results in a lower error rate than for the top-hat
encoding. We conclude that Gaussian modulation can allow for an
improved trade-off between encoding, processing and measurement of
the information.
\end{abstract}
\pacs{03.67.Ac}
\maketitle

%==========================================================================%
%%%%%%%%%%%%%%%%%%%%%%%%%%%%%%%%%%%%%%%%%%%%%%%%%%%%%%%%%%%%%%%%%%%%%%%%%%%%%%%%%%%
%%%%%%%%%%%%%%%%%%%%%%%%%%%%%%%%%%%%%%%%%%%%%%%%%%%%%%%%%%%%%%%%%%%%%%%%%%%%%%%%%%%
\section{Introduction}

The process of computation involves taking an input and converting
it into an output through the application of an
algorithm~\cite{CLRS01, Si97}. There are two versions of classical
computation: digital and analogue. Digital computers use
finite-length bit strings as the input and output of the computer.
Analogue computers typically use the electrical quantities of
inductance and charge as analogues of mass and displacement so
that continuously varying voltages may output, for example, the
simulated trajectory of a spacecraft~\cite{Ja74}.

Quantum computation also has two versions of information
processing referred to as discrete-variable~\cite{NC00} and
continuous-variable~\cite{BP03} quantum computation. In the
discrete-variable version, algorithms employ a finite numbers of
qubits~\cite{NC00} or qudits~\cite{GKP01} analogously to bits in
the classical digital case.  Continuous-variable quantum
computation uses continuously-parameterized quantum systems to
process discrete information.

Comparing the performance of a well-known problem and its
algorithm in both discrete and continuous settings provides useful
insight into how the different versions of quantum computation
differ in their performance. The historically important
Deutsch--Jozsa problem~\cite{De85,DJ92} has been used to show the
different algorithm performance that can be achieved in the two
settings. In the discrete variable setting, the Deutsch--Jozsa
problem can solved exactly with a single query of a quantum
oracle~\cite{DJ92}. In the continuous variable setting, quantum
algorithms employing orthogonal wave functions as the
computational basis necessarily have single-query success
probability less than one~\cite{AHS09}.

Continuous-variable studies are often based on quantum optics
because of the wide variety of tools that have been developed to
process and measure optical field
modes~\cite{BP03,FSB98,GG02,AFK08,AYA07,Br98,EP02,BSBN02,BS02}.
Continuous-variable quantum algorithms have been studied that use
a single mode~\cite{AHS09} and multiple modes~\cite{CHMS10}. In
continuous-variable quantum information procedures, the input
state is typically a Gaussian wave function over the canonical
position representation with the physical system being a harmonic
oscillator (equivalently a single-mode optical field~\cite{Le97}).
However the unbounded extent of these wave functions does not
naturally fit with the finite length of the information strings
being processed.

To deal with the problem of finite-length information strings, the
single-mode algorithm defined in~\cite{AHS09} has input states
represented by orthogonal wave functions. Information is
represented by finite-length bit strings $z \in\{0,1\}^N$ with $N$
the number of bits. These strings are encoded into a region of the
momentum domain $p$ extending from $-P$ to $P$. A regular lattice
of $N$ discrete values of $p$ are embedded in this domain such
that the $i^{\rm{th}}$ bit of $z$, is assigned to $p$ as follows:
$p_{z_i} := (1+2i-N\pm 1)P/N$. The lattice is thus $\{p_{z_i}\}$
with spacing $2P/N$.

An orthogonal basis of top-hat functions is formed from the
$p_{z_i}$, and the encoded momentum wave function is expressed as
$\frac{1}{\sqrt{2P}}\sum_{i=0}^{N-1}(-1)^{z_i}\left|p_{z_i}\right\rangle$.
The kets $\left|p_{z_i}\right\rangle$ are phase-modulated by their
corresponding bit values, and with this phase modulation, each of
the possible $2^N$ strings is uniquely represented. Note that the
constant wave function is the top-hat function extending from $-P$
to $P$.

There is a translational invariance between computational basis
states because each of the top-hat basis functions occupies an
identically-sized region of momentum space and because the string
$z'$ is obtained from the string $z$ by the translation
$z'=z\oplus(z\oplus z')$. We can regard this finite basis as an
infinite basis modulated by a top-hat function extending from $-P$
to $P$, which has the effect of truncating the allowed strings
from an infinite domain to being from $0$ to~$N-1$.

In this paper, we extend this approach by shaping the overall
top-hat function with a Gaussian having width set by its standard
deviation $\sigma$. The Gaussian is truncated for $|p|>P$, which
has the effect that the computational basis states, although still
orthogonal, are no longer translations of each other but are more
complicated Gaussian-modulated basis states. At first glance, this
latter feature would appear to provide a disadvantage, but using
the mathematical properties of the error function, we prove that
the single-query success probability for the Gaussian case
$\Pr_\checkmark^\sharp > \Pr_\checkmark^\bot$ is greater than the
single-query success probability for the orthogonal case.

This result is enabled by the extra degree of freedom manifest in
the spread of the Gaussian wave function. Tuning the available
parameters of encoding width, the spread of the Gaussian wave
function and the width of the measurement window results in a more
efficiently encoded momentum wave function leading to the improved
single-query success probability. Note that the use of the top-hat
basis to encode information into a single harmonic oscillator is
different than the approach used in~\cite{GKP01}, where
information is encoded into a collection of harmonic oscillators.

Our paper is organized as follows. In Sec.~\ref{sec:background},
we define oracle-decision problems and present the Deutsch--Jozsa
~\cite{DJ92,De85,CEMM98} in this context. We give an overview of
the single-mode algorithm employing orthogonal states~\cite{AHS09}
that solves the Deutsch--Jozsa problem. We give a brief
introduction into the coherent states of the harmonic oscillator
and define the single-mode algorithm in these terms. In Sec.\
\ref{sec:sharpcutoff}, we prove that the single-query success
probability claimed for the Gaussian model with truncated momentum
is better than that achieved using orthogonal states~\cite{AHS09}.
We conclude in Sec.\ \ref{sec:conclusions}.

%%%%%%%%%%%%%%%%%%%%%%%%%%%%%%%%%%%%%%%%%%%%%%%%%%%%%%%%%%%%%%%%%%%%%%%%%%%%%%%%%%%%%%%%%%%%%%%%%%%%%%%%%%%%%%%%%%%%%%%%%%%%%%
%Section 2%
\section{Background}\label{sec:background}

An important aspect of quantum information processing is the
ability to solve oracle decision problems with increased
efficiency compared to classical information processing. In the
case of oracle decision problems, efficiency is measured in terms
of the number of oracle queries required to solve the problem.
Comparing the single-query success probability of an algorithm in
both the discrete and continuous quantum settings provides useful
insights into the advantages of one setting over the other. The
Deutsch--Jozsa oracle-decision problem has been studied in both
the discrete and continuous-variable quantum settings
~\cite{DJ92,BP03,AHS09,AHS11}. Here we explore a single-mode
continuous-variable quantum algorithm where the input state is a
Gaussian wave function.

\subsection{Oracle Decision Problems and the Deutsch--Jozsa Problem}

The challenge of an oracle decision problem is to identify which
of two mutually disjoint sets contains a unique $N$-bit string by
making the fewest possible queries to an oracle. The oracle
decision problem is typically couched in terms of a function $f$
that maps $N=2^n$-bit strings to a single bit
\begin{align}\label{classical_gl1}
f:\{0,1\}^n\mapsto\{0,1\}.
\end{align} Any Boolean function on $n$ bits can also be represented by a
string of $N=2^n$ bits, in which the $i^{\rm{th}}$ bit $z _i$ is
the value of the function on the $i^{\rm{th}}$ bit string, taken
in lexicographical order.

For our analysis, we repeat the definition of an oracle decision
problem given in~\cite{AHS11} as follows.
\begin{definition}\label{def:1}
An oracle decision problem is specified by two non-empty, disjoint
subsets $A,B \subset \{0,1\}^N$. Given a string $z \in A \cup
B=C$, the oracle-decision problem is to determine whether $z\in A$
or $z\in B$ with the fewest queries to the oracle possible.
\end{definition}

For completeness, we also repeat the definition of the
Deutsch--Jozsa problem~\cite{DJ92,CEMM98} in terms of
Definition~\ref{def:1} given in~\cite{AHS11} as follows.
\begin{problem}\sl\label{DJProblem}
Given set the set of balanced strings $A\subset\{0,1\}^N$, where
exactly $N/2$ elements take on the value 0 and the set of constant
strings $B\subset\{0,1\}^N$, where all $N$ elements take on the
same value everywhere, and a string $z$ randomly selected with
uniform distribution~$\mu$ such that $z\in_{\mu} C=A \cup B$, the
Deutsch--Jozsa Problem is to determine if $z\in A$ or $z\in B$
with the fewest oracle queries.
\end{problem}

\subsection{Single Mode Continuous Variable Algorithm with Orthogonal States}

In Fig.~\ref{fig:GVDJDiag}, we present the single-mode,
continuous-variable quantum algorithm~\cite{AHS09} for the
solution of oracle decision problems. The vertical lines on
Fig.~\ref{fig:GVDJDiag} represent the states after the various
steps of the algorithm using function notation rather than Dirac
notation. In function notation, the Dirac ket $|\phi\rangle$ is
represented by the square-integrable function
\begin{align}
\phi(x)=\left\langle x|\phi\right\rangle,
\end{align} where $x$ in this case is the continuous position variable.

The square-integrable condition means that orthogonal functions
may be used to represent the wave functions. One possible set of
orthogonal functions is the Fourier-transform pair realized by the
sinc/top-hat functions~\cite{AHS09}. In this case, the sinc
function
\begin{align}
    \phi_0(x)=\frac{\sin(P x)}{\sqrt{\pi P}\, x}
\end{align}
is the input state, and its Fourier transform is the momentum
top-hat function
\begin{equation} \label{eq:puremomentum}
   \tilde{\phi}(p)
        =\left(\frac{1}{\sqrt{2P}}\right)
            \left\{
                \begin{array}{ll}1,&\text{if}\,\,p\in[-P,P]\\0,&\text{if}\,\,p\notin[-P,P],\end{array}
            \right.
\end{equation} having finite
extent of $2P$ in the momentum domain. One nice feature of the
sinc/top-hat pair is that the finite extent of the top-hat
distribution allows for finite length information to be encoded
naturally.

\begin{figure}[tbp]
            \begin{center}
            \includegraphics[width=9 cm]{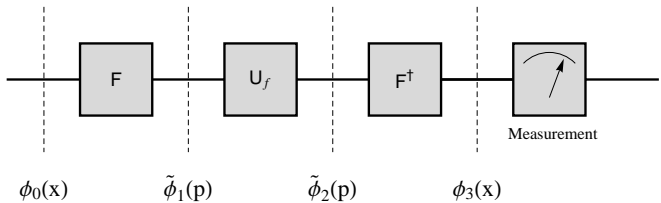}
            \end{center}
            \caption{Single-mode quantum circuit implementing the continuous-variable Deutsch--Jozsa algorithm~\cite{AHS09}.} \label{fig:GVDJDiag}
\end{figure}

The encoded position sinc function has unbounded extent, and
analysis of the optimum position measurement window reveals an
uncertainly relationship~\cite{AHS09} between the measurement
window $\delta$ and the encoding length $P$ expressed as
\begin{align}
P\delta =\frac{\pi}{2}.
\end{align}
As a result, the single-mode continuous variable algorithm is
necessarily probabilistic~\cite{AHS09} and has single-query
success probability
\begin{align}\label{Eq:orthogSucProb}
\text{Pr}^{\bot}_{\checkmark}=0.61.
\end{align} Here we demonstrate that this single query success probability may be improved
upon by using a Gaussian wave function as algorithm input.

\subsection{Single Mode Continuous Variable Algorithm with Gaussian States}

The sinc function employed as algorithm input in~\cite{AHS09}
cannot be readily created in the laboratory. Here we are inspired
by the ability to create and manipulate physical states of light
in the laboratory using the tools of quantum optics. In
particular, we employ coherent states, which may by represented by
Gaussian wave functions, as the input states to our algorithm.

The method of coherent states is well established and one feature
is that the coherent states are overcomplete~\cite{Ar72,Pe72}.
High quality lasers generate light fields that are
coherent~\cite{Le97}. The vacuum state is a displaced coherent
state and as such has the same quantum noise properties. Coherent
states are usually expressed as the ket
$\left|\alpha\right\rangle$ with $\alpha=x_0+{\rm{i}}p_0$ to
reflect that it is a state that is shifted from the vacuum by the
magnitude $|\alpha|$. Similarly, the vacuum is usually expressed
simply as the ket~$|0\rangle$.

In the position representation, the coherent state of laser light
may be expressed as
\begin{align}
\label{eq:testftnnotation} \phi(x;\alpha)=\left\langle x | \alpha
\right\rangle =\frac{e^{-\frac{1}{2}\left[(x-x_0)^2 - 2{\rm{i}}p_0
x + {\rm{i}}p_0 x_0\right]}}{\sqrt[4]{\pi }},
\end{align} where~$x_0=p_0=0$ corresponds to the vacuum state.
From the perspective of our quantum algorithm, the displaced
vacuum behaves no differently than the vacuum itself. Therefore
for notational simplicity, we chose to use the position
representation of the vacuum,
\begin{align}
\label{eq:testftnnotation} \left\langle x | 0 \right\rangle
=\frac{e^{-\frac{x^2}{2}} }{\sqrt[4]{\pi }},
\end{align} as the starting-point state for our algorithm.

Quantum optics has many tools that allow for the manipulation of
light. Of interest in our algorithm is light squeezing, where
quantum uncertainties are redistributed altering the shape of the
distribution. The squeezing operator is given in \cite{Le97} as
\begin{equation}
\hat{S}(\zeta)=\exp\left(\frac{\zeta}{2}\left(\hat{a}^2-\hat{a}^{\dag
2}\right)\right),
\end{equation}
where $\hat{a}=\hat{x}+{\rm{i}}\hat{p}$ is the annihilation
operator and $\hat{a}^{\dag}=\hat{x}-{\rm{i}}\hat{p}$ is the
creation operator. The quadratures $\hat{x}$ and $\hat{p}$ are
regarded as the position and momentum of the harmonic oscillator,
and the quantity $\zeta$ is referred to as the squeezing
parameter~\cite{Le97}.

In Dirac notation, the squeezed vacuum state may be expressed as
\begin{align}
|\phi_0\rangle=\hat{S}(\zeta)|0\rangle.
\end{align}
In our analysis, we use the standard deviation $\sigma=e^{-\zeta}$
to represent the effect of the squeezing operator on our function
representation of a Gaussian state.

We employ the squeezed vacuum as the input state to our algorithm,
which we represent in function notation as
\begin{align}
\label{Eq:AlgInput} \phi_0(x;\sigma)&=\langle x |\hat{S}(\zeta)|0\rangle\nonumber\\
&=\frac{e^{-\frac{x^2}{2 \sigma ^2}} }{\sqrt[4]{\pi }\sqrt{\sigma
}}.
\end{align} The subscript zero identifies this state as the algorithm input state
represented by the leftmost vertical line in
Fig.~(\ref{fig:GVDJDiag}). We prove that the algorithm with this
Gaussian input state has improved single-query success probability
over the algorithm employing orthogonal states as input.

\begin{theorem}\label{sharpclaim} Using the single-mode quantum circuit given in
Fig.~\ref{fig:GVDJDiag} with the coherent state given by
Eq.~(\ref{Eq:AlgInput}) as algorithm input and employing sharp
information cutoff, the single-query success probability
$\Pr^\sharp_\checkmark>\Pr^ \bot_\checkmark$ is greater than the
single-query success probability $\text{Pr}^{\bot}_{\checkmark}$
obtained using
orthogonal states given by Eq.~(\ref{Eq:orthogSucProb}).\\
\end{theorem}

%Section 2%
\section{Bounding the Query Complexity of the Single Mode Algorithm with Gaussian Input States}
\label{sec:sharpcutoff}

The continuous-variable quantum algorithm using orthogonal states
solves Problem 1 with exponentially small error probability
$1/(2^m)$ in a linear number of queries~$\Theta(m)$~\cite{AHS09}.
This query complexity is dependent on the single-query success
probability~$\Pr^ \bot_\checkmark$, which is a measure of the
maximum achievable separation between the probability that the
encoded string is a balanced string versus a constant string.
Here, where the input is a Gaussian state, we demonstrate that the
key parameters affecting this separation are the encoding width,
the spread of the Gaussian wave function and the width of the
measurement window. We vary these parameters and discover their
optimum values in our proof of Theorem~1.

\subsection{Encoding Information into Gaussian States}

With reference to Fig.~\ref{fig:GVDJDiag}, the first step of the
algorithm is to take the Fourier transform~\cite{AHS09} of the
input state  $\phi_0(x;\sigma)$ giving
\begin{align}
\label{eq:unencodedpgaussian}
\tilde{\phi}_1(p\,;\sigma)=\frac{e^{-\frac{1}{2} p^2 \sigma ^2}
\sqrt{\sigma }}{\sqrt[4]{\pi }}.\end{align} The next step has the
oracle $\op{U}_f$ modulate the momentum Gaussian with the pulse
train that represents the encoding of the $N$-bit string $z$.

The modulated momentum wave function is
\begin{align}
\label{eq:eEncodedMomentumGaussian}
    \tilde{\phi}_2(p)=\tilde{\phi}_z^{(N)}(p\,;\sigma,P)=\eta f_z^{(N)}(p;P)\tilde{\phi}_1(p\,;\sigma),
\end{align} where we have labelled the state with all relevant
parameters. Descriptions of the elements of this equation follow.
The modulating square-wave encoded with the $N$-bit string $z$ is
\begin{align}
    f_z^{(N)}(p;P)=\sum_{i=0}^{N-1} (-1)^{z_i}\sqcap_i^{(N)}(p;P),
\end{align}
where the definition of the momentum bins given in~\cite{AHS09} is
repeated here as
\begin{align}
    \sqcap_i^{(N)}(p;P)=\left\{
            \begin{array}{ll}
                1,&\frac{p}{P}\in\left[-\left(1-2\frac{N-1-i}{N}\right),
                    -\left(1-2\frac{N-i}{N}\right)\right]\\0,&\text{otherwise.}
            \end{array}
        \right.
\end{align}
Note that the modulating function has the effect of chopping off
the tails of the momentum Gaussian outside $\pm P$ thus truncating
the Hilbert space.

The normalization factor, $\eta$, of the chopped distribution is
calculated as
\begin{align}
\int_{-\infty}^{\infty}\left|\tilde{\phi}_z^{(N)}(p;\sigma,
P)\right|^2\mathrm{d}p =\text{erf}(P\sigma),\end{align} where the
error function is
$\text{erf}(w)=\frac{2}{\sqrt{\pi}}\int_{0}^{w}e^{-t^2}dt$, and
\begin{align}
\label{eq:normalizationfactor}
    \eta=1/\sqrt{\text{erf}(P\sigma)}.
\end{align}
The penultimate step is to take the inverse Fourier transform of
this encoded momentum state.

The encoded position state is thus expressed as
\begin{align}\label{Eq:ModulatedPositionWaveFunctionPhi3(x)}
\phi_{3}(x)=\phi_{z}^{(N)}(x;\sigma,P)=\frac{\eta\,\text{e}^{-\frac{x^2}{2\sigma^2}}}{2\sqrt[4]{\pi}\sqrt{\sigma}}
\,M^{(N)}_z(x;\sigma,P).
\end{align} The effect of the encoded information is
completely captured in the position modulating term
\begin{align}
M^{(N)}_z(x;\sigma,P)=&\sum_{j=1}^{N}(-1)^{z_j}
\left[\text{erf}\left(\frac{\vartheta_j \sigma^2+{\rm i} x}{\sqrt{2}\sigma}\right)\right.\nonumber\\
&-\left.\text{erf}\left(\frac{\vartheta_{j-1} \sigma^2+{\rm i}
x}{\sqrt{2}\sigma}\right)\right], \label{Eq:PosnModTerm}
\end{align}
where
\begin{align}\label{Eq:VarThetaModulation}
\vartheta_{j}=P\left(\frac{2j-N}{N}\right).\end{align} The final
step is the measurement step.

We follow the same approach taken in~\cite{AHS09} and calculate
the probability of detecting a particular wave function in the
interval $\pm\delta$ as
\begin{align}\label{Eq:ProbabilityDef}
\Pr\left[\phi_{z}^{(N)}(x;\sigma,P)\right]=\int_{-\delta}^{\delta}\left
|\phi_{z}^{(N)}(x;\sigma,P)\right|^2\text{d}x.
\end{align} Since the wave function may be encoded with a constant string or a
balanced string, we need to determine the optimal value of
$\delta$ that maximizes our ability to distinguish between these
cases. Our approach is to determine which balanced functions
dominate all other balanced functions in the measurement window.

We begin by defining three pairs of $N$-bit strings: the
antisymmetric balanced (AB) strings, the symmetric balanced (SB)
strings and the constant (C) strings as
\begin{align}
\text{\small{AB}}&\in\left\{\underbrace{0\cdots
0}_{N/2}\underbrace{1\cdots 1}_{N/2}, \underbrace{1\cdots
1}_{N/2}\underbrace{0\cdots 0}_{N/2}\right\},\label{eq:string_AB}\\
\text{\small{SB}}&\in\left\{\underbrace{0\cdots
0}_{N/4}\underbrace{1\cdots 1}_{N/2}\underbrace{0\cdots 0}_{N/4},
\underbrace{1\cdots 1}_{N/4}\underbrace{0\cdots
0}_{N/2}\underbrace{1\cdots 1}_{N/4}\right\},\label{eq:string_SB}\\
 \text{\small{C}}&\in\left\{\underbrace{0\cdots 0}_{N},
\underbrace{1\cdots 1}_{N}\right\}. \label{eq:string_C}
\end{align} Note that the constant stings have zero bit
transitions, the antisymmetric balanced strings have one bit
transition, and the symmetric balanced strings have two
transitions. All other balanced strings have two or greater
transitions.

It is insightful to analyze the modulating term given by
Eq.~(\ref{Eq:PosnModTerm}) for $x=0$, which we express as
\begin{align}
M^{(N)}_z(0;\sigma,P)&=\sum_{j=1}^{N}(-1)^{z_j}
\left[\text{erf}\left(\frac{\vartheta_j \sigma}{\sqrt{2}}\right)
-\text{erf}\left(\frac{\vartheta_{j-1} \sigma
}{\sqrt{2}}\right)\right]. \label{Eq:ErfSumZero}
\end{align}We use the anti-symmetric property of the error function
$\text{erf}(a)=-\text{erf}(-a)$, and the property that
$\text{erf}(0)=0$. We also use the facts that $\vartheta_N=P$,
$\vartheta_N/2=0$ and $\vartheta_j=-\vartheta_{N-j}$
for~$j=0,1,\ldots,N$ in determining the following results.

For the constant case, all the terms cancel except the first and
last, and we have
\begin{align}
M^{(N)}_{z\in \text{C}}(0;\sigma,P)&=\pm
\left[\text{erf}\left(\frac{\vartheta_N \sigma}{\sqrt{2}}\right)
-\text{erf}\left(\frac{\vartheta_{0}
\sigma }{\sqrt{2}}\right)\right]\nonumber\\
&=\pm 2\,\text{erf}\left(\frac{P\sigma}{\sqrt{2}}\right).
\label{Eq:ErfSumZeroConst}
\end{align}
For the antisymmetric balanced case, all the terms cancel and we
have
\begin{align}
M^{(N)}_{z\in AB}(0;\sigma,P)&=\pm
\left[\text{erf}\left(\frac{\vartheta_N \sigma}{\sqrt{2}}\right)
+\text{erf}\left(\frac{\vartheta_{0}
\sigma }{\sqrt{2}}\right)\right]\nonumber\\
&=0. \label{Eq:ErfSumZeroASB}
\end{align}
For the symmetric case we have
\begin{align}
M^{(N)}_{z\in SB}(0;\sigma,P)&=\pm
\left[2\,\text{erf}\left(\frac{P\sigma}{\sqrt{2}}\right)
-4\,\text{erf}\left(\frac{P \sigma }{2\sqrt{2}}\right)\right].
\label{Eq:ErfSumZeroSym}
\end{align} Here the sum is non-zero except for in the
limiting case where $P\sigma\rightarrow 0$.

It is apparent from the results of Eqs.~(\ref{Eq:ErfSumZeroASB})
and~(\ref{Eq:ErfSumZeroSym}) that there are different classes of
balanced functions. Some balanced functions are only non-zero at
$x=0$ in the limit as $\sigma$ goes to zero, and some balanced
functions are zero at $x=0$ for all values of $\sigma$. We need to
determine which balanced functions are the two functions that
dominate the measurement region.

\begin{figure}[tbp]
            \begin{center}
            \includegraphics[width=9 cm]{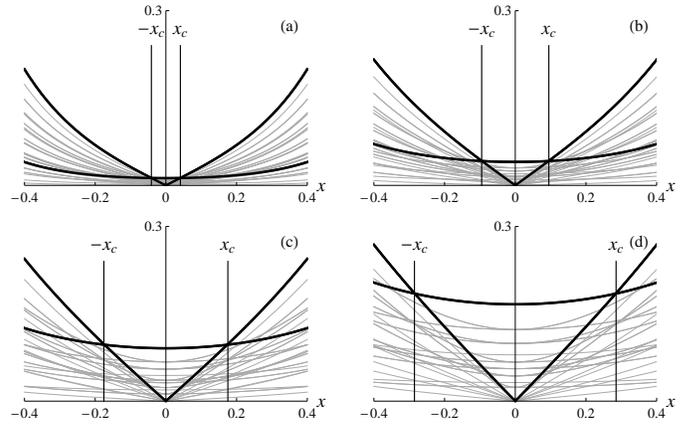}
            \end{center}
            \caption{Plots of the magnitude of the position modulation function $\left|M^{(8)}_z(x;\sigma,1)\right|$ given
            by Eq.~(\ref{Eq:ModulatedPositionWaveFunctionPhi3(x)}) for (a) $\sigma=0.4$, (b) $\sigma=0.6$,
            (c) $\sigma=0.8$, and (d) $\sigma=1.0$. For $\left|x\right|<x_c$, the symmetric balanced function (bold) dominates all other balanced functions,
            and for $\left|x\right|>x_c$, the antisymmetric balanced
            function (bold)
            dominates all other balanced functions.} \label{fig:DominateBalancedFuntions}
\end{figure}

In Appendix A, we prove that the magnitude of the position
modulation function $\left|M^{(N)}_z(x;\sigma,1)\right|$ given by
Eq.~(\ref{Eq:ModulatedPositionWaveFunctionPhi3(x)}) and subject to
the balanced condition $\sum_{i=1}^{N} (-1)^{z_i}=0$, is maximized
by either the antisymmetric balanced function given by
Eq.~(\ref{eq:string_AB}) or the symmetric balanced function given
by Eq.~(\ref{eq:string_SB}). For $N=8$, the situation is presented
in Fig.~\ref{fig:DominateBalancedFuntions}, where it be seen that
the actual dominating function is dependent on the value
of~$\sigma$.

In Fig.~\ref{fig:DominateBalancedFuntions}, the crossover point is
drawn and is approximated in Appendix \ref{App:CalcDet} as
\begin{align}
x_c\approx \frac{P\sigma^2}{(4 - P^2\sigma^2)}. \label{Eq:XcDef}
\end{align}
For $|x|<x_c$, the symmetric balanced function (shown in bold)
dominates, and for $|x|>x_c$, the antisymmetric balanced function
(shown in bold) dominates. All remaining balanced functions, of
which there are a total of $\binom{8}{4}=70$ are depicted as light
gray lines in Fig.~\ref{fig:DominateBalancedFuntions}. We use
these results to complete the proof of Theorem~1.

\subsection{Proof of Theorem~1}

We need to maximize the separation between detecting a balanced
string and a constant string. To this end, we define the following
quantities
\begin{align}\label{Eq:DeltaAB}
\Delta_{\text{AB}}(\sigma,P,\delta)=\left|\Pr\left[\phi_{z\in\text{C}}^{(N)}\right]
-\Pr\left[\phi_{z\in\text{AB}}^{(N)}\right] \right|,
\end{align} and
\begin{align}\label{Eq:DeltaSB}
\Delta_{\text{SB}}(\sigma,P,\delta)=\left|\Pr\left[\phi_{z\in\text{C}}^{(N)}\right]
-\Pr\left[\phi_{z\in\text{SB}}^{(N)}\right] \right|,
\end{align} where for brevity, we have suppressed the arguments in
Eq.~(\ref{Eq:ProbabilityDef}). The single-query success
probability is defined in these terms as
\begin{align}
\text{Pr}_{\checkmark}^{\sharp}=\min\left[\Delta_{\text{AB}}(\sigma,P,\delta),
\Delta_{\text{SB}}(\sigma,P,\delta)\right].
\end{align}
This expression assumes that either the antisymmetric or the
symmetric balanced strings dominate all other balanced strings in
the region $\pm\delta$ as presented in the previous subsection and
proved in Appendix A. We seek to determine the values of $\delta$
and $\sigma$ that maximize the separation between these two
probabilities.

With $\Delta_{\text{AB}}(\delta,\sigma,P)$ defined in
Eq.~(\ref{Eq:DeltaAB}), we
set~$\Delta_{\text{AB}}'(\delta,\sigma,P)=\frac{\partial}{\partial\delta}\Delta_{\text{AB}}(\delta,\sigma,P)$,
 which gives us
\begin{align}
\label{eq:findoptimumdelta}
&\Delta_{\text{AB}}'(\delta,\sigma,P)=\nonumber\\
&\frac{e^{-\frac{\delta^2}{\sigma^2}}}{2\sqrt{\pi}\sigma\,
\text{erf}(P\sigma)}\left[\text{erf}\left(\frac{P \sigma
^2-{\rm{i}}\delta}{\sqrt{2}\sigma}\right)^2+\text{erf}\left(\frac{P
\sigma ^2+{\rm{i}}\delta}{\sqrt{2}\sigma}\right)^2
\right.\nonumber\\
&\left.+\,2\,\text{erf}\left(\frac{P \sigma
^2-{\rm{i}}\delta}{\sqrt{2}\sigma}\right)\text{erf}\left(\frac{{\rm{i}}\,\delta}{\sqrt{2}
\sigma}\right)+2\,\text{erf}\left(\frac{{\rm{i}}\,\delta}{\sqrt{2}
\sigma}\right)^2
\right.\nonumber\\
&\left.+\,2\,\text{erf}\left(\frac{P \sigma
^2+{\rm{i}}\delta}{\sqrt{2}\sigma}\right)\text{erf}\left(\frac{{\rm{i}}\,\delta}{\sqrt{2}
\sigma}\right) \right].
 \end{align}
It suffices to set $\Delta_{\text{AB}}'(\delta,\sigma,P)=0$ to
maximize the separation. Before doing so, we elect to simplify
Eq.~(\ref{eq:findoptimumdelta}) by `normalizing' the standard
deviation $\sigma $ and the measurement `length' $\delta$ with
respect to the encoding `length'~$P$.

We assume that the uncertainty relation~\cite{AHS09} remains true
up to a constant. We express this as
\begin{equation}
\label{eq:deltahat} P \delta =\bar{\delta}.
\end{equation} This assumption and analysis of the error function arguments of
Eq.~(\ref{eq:findoptimumdelta}) result in a similar uncertainty
relationship between $P$ and $\sigma$, which we express as
\begin{equation}
\label{eq:sigma hat} P\sigma=\bar{\sigma}.
\end{equation} Making the substitutions given by Eq.~(\ref{eq:deltahat}) and
Eq.~(\ref{eq:sigma hat}) into Eq.~(\ref{eq:findoptimumdelta}) and
setting it to 0 results in the following expression
\begin{align}
\label{eq:findoptimumdelta_var_ab}
&\Delta_{\text{AB}}'\left(\bar{\delta},\bar{\sigma}\right)=0\nonumber\\
&=\left[\text{erf}\left(\frac{\bar{\sigma}^2-{\rm{i}}\,\bar{\delta}}{
\sqrt{2}\bar{\sigma}}\right)^2+\text{erf}\left(\frac{\bar{\sigma}^2+{\rm{i}}\,\bar{\delta}}{\sqrt{2}\bar{\sigma}}\right)^2
\right.\nonumber\\
&\left.+\,2\,\text{erf}\left(\frac{\bar{\sigma}^2-{\rm{i}}\,\bar{\delta}}{
\sqrt{2}\bar{\sigma}}\right)\text{erf}\left(\frac{{\rm{i}}\,\bar{\delta}}{\sqrt{2}
\bar{\sigma}}\right)+2\,\text{erf}\left(\frac{{\rm{i}}\,\bar{\delta}}{\sqrt{2}
\bar{\sigma}}\right)^2
\right.\nonumber\\
&\left.+\,2\,\text{erf}\left(\frac{\bar{\sigma}^2+{\rm{i}}\,\bar{\delta}}{\sqrt{2}\bar{\sigma}}\right)
\text{erf}\left(\frac{{\rm{i}}\,\bar{\delta}}{\sqrt{2}
\bar{\sigma}}\right) \right].
 \end{align} Note that the variables $\bar{\sigma}$ and $\bar{\delta}$ are, in some sense the
`normalized' Gaussian standard deviation $\sigma$ and the
measurement width $\delta$, `scaled' by the momentum `length' $P$.

Eq.~(\ref{eq:findoptimumdelta_var_ab}) is dependent on the two
variables, $\bar{\delta}$ and $\bar{\sigma}$ and is thus
insufficient to find the global optimum values of $\bar{\sigma}$
and $\bar{\delta}$. We obtain the needed constraint from the
similar equation derived from the symmetric balanced function.
Following the same steps we did in Eq.~(\ref{eq:findoptimumdelta})
and Eq.~(\ref{eq:findoptimumdelta_var_ab}), we obtain the
following expression
\begin{align}
\label{eq:findoptimumdelta_var_sym}
&\Delta_{\text{\tiny{SB}}}'\left(\bar{\delta},\bar{\sigma}\right)=0\nonumber\\
&=\left[\text{erf}\left(\frac{\bar{\sigma}^2-i \bar{\delta}}{2
\sqrt{2}
\bar{\sigma}}\right)+\text{erf}\left(\frac{\bar{\sigma}^2+i
\bar{\delta}}{2 \sqrt{2}
 \bar{\sigma}}\right)\right]\nonumber\\
  &\times \left[-\text{erf}\left(\frac{^2-i \bar{\delta}}{2 \sqrt{2} \bar{\sigma}}\right)
  +\text{erf}\left(\frac{\bar{\sigma}^2-i \bar{\delta}}{\sqrt{2}
   \bar{\sigma}}\right)\right.\nonumber\\
   &-\left.\text{erf}\left(\frac{\bar{\sigma}^2+i \bar{\delta}}{2 \sqrt{2} \bar{\sigma}}\right)
   +\text{erf}\left(\frac{\bar{\sigma}^2+i \bar{\delta}}{\sqrt{2} \bar{\sigma}}\right)\right].
 \end{align} We solve Eqs.~(\ref{eq:findoptimumdelta_var_ab}) and~(\ref{eq:findoptimumdelta_var_sym}) simultaneously to
establish the optimum values of the measurement lengths
$\bar{\delta}_{\text{AB}}$ and $\bar{\delta}_{\text{SB}}$ in terms
of the normalized standard deviation $\bar{\sigma}$.

\begin{figure}[tbp]
            \begin{center}
            \includegraphics[width=9 cm]{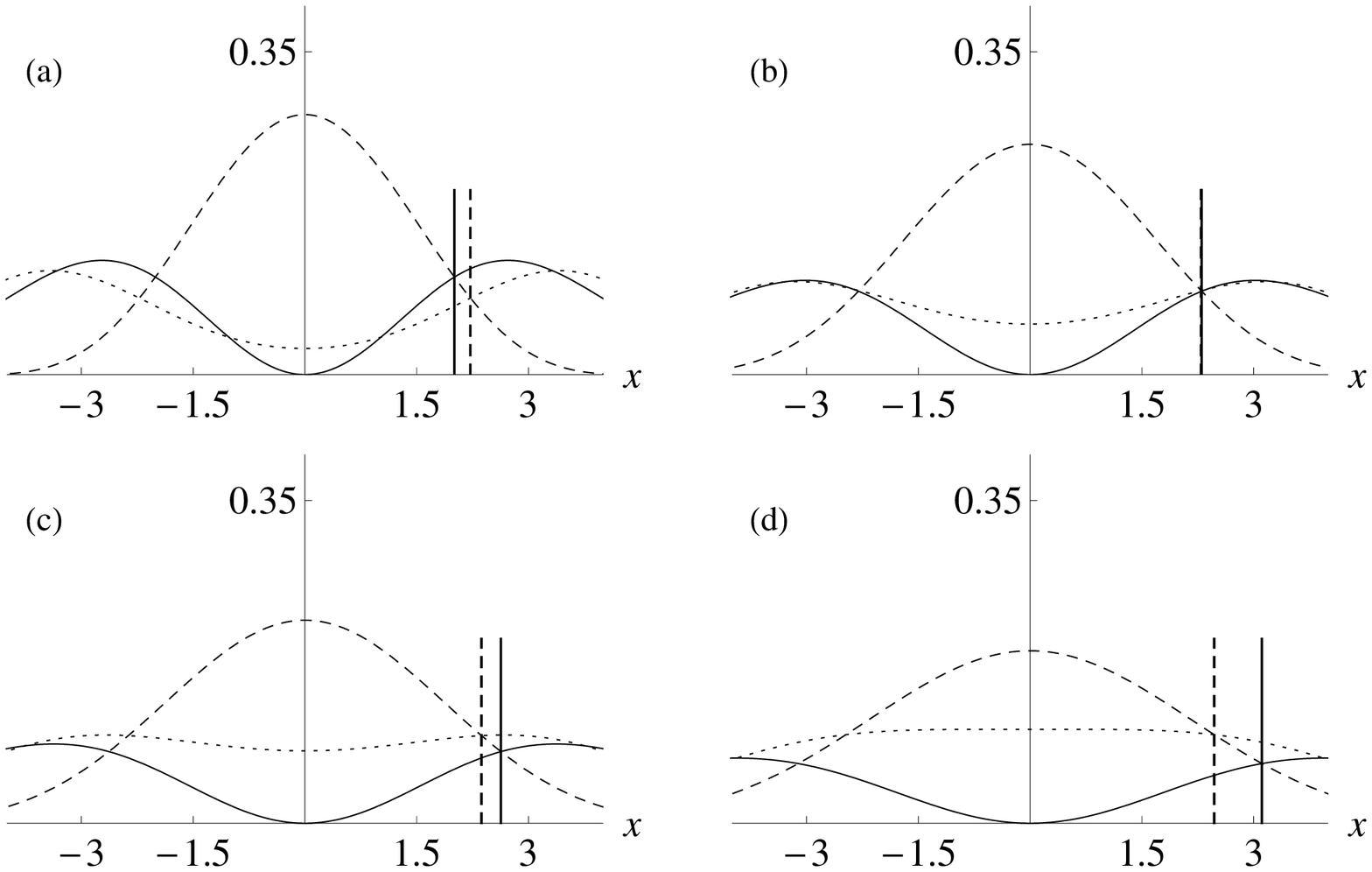}
            \end{center}
            \caption{Plots of $\left|\phi^{(N)}_z(x;\sigma,1)\right|^2$ for~$ z \in
\{\text{AB},\text{SB},\text{C}\}$ (solid, dotted, dashed
respectively) for: (a)~$\sigma=1.67$, (b)~$\sigma=2.11$,
(c)~$\sigma=2.5$, and (d)~$\sigma=3.0$. The solid vertical lines
in all four sub-plots correspond to
$\Delta'_{\text{\tiny{AB}}}\left(\bar{\delta},\bar{\sigma}\right)=0$,
and the dashed vertical lines correspond to
$\Delta'_{\text{\tiny{SB}}}\left(\bar{\delta},\bar{\sigma}\right)=0$.
Note that a) corresponds to the values that optimize the
single-query success probability. } \label{fig:GVDJoptimalsigmaP2}
\end{figure}

In Fig.~\ref{fig:GVDJoptimalsigmaP2}, we plot the distributions
$|\phi^{(N)}_z(x;\sigma,P)|^2$ for
$z\in\{\text{AB},\text{SB},\text{C}\}$ for several values of
$\sigma$.  We also plot vertical lines corresponding to the values
of $\bar{\delta}$ corresponding to
$\Delta^{'}_{\text{\tiny{AB}}}=0$ and
$\Delta^{'}_{\text{\tiny{SB}}}=0$. Note that there are values of
the normalized parameters~$\bar{\sigma}$ and $\bar{\delta}$ where
$\Delta^{'}_{\text{\tiny{AB}}}\left(\bar{\delta},\bar{\sigma}\right)=0$,
and
$\Delta^{'}_{\text{\tiny{SB}}}\left(\bar{\delta},\bar{\sigma}\right)=0$
simultaneously. This situation occurs where~$\bar{\delta}\approx
2.30$ and~$\bar{\sigma}\approx 2.11$ and is depicted in
Fig.~\ref{fig:GVDJoptimalsigmaP2}(b).

However, these values do not optimize the success probability
since
\begin{align}
&\Delta_{\text{\tiny{AB}}}(2.30,2.11)\approx 0.68,\,\, \text{and}\nonumber\\
&\Delta_{\text{\tiny{SB}}}(2.30,2.11)\approx
0.54.\label{Eq:Deltapeqzero}
\end{align} Lack of
optimality is manifest in the lower of the two above values, which
is less than single-query success probability for the orthogonal
case $\text{Pr}_{\checkmark}^{\bot}\approx 0.61$.

Increasing the value of~$\bar{\sigma}$ further serves to
increase~$\Delta_{\text{\tiny{AB}}}\left(\bar{\delta},\bar{\sigma}\right)$
and
decrease~$\Delta_{\text{\tiny{SB}}}\left(\bar{\delta},\bar{\sigma}\right)$,
which worsens the success probability. Reducing the value
of~$\bar{\sigma}$ brings them together. The quantity
$\Delta_{\text{\tiny{AB}}}\left(\bar{\delta},\bar{\sigma}\right)$
thus takes on its maximum value when
\begin{align}
\Delta'_{\text{\tiny{AB}}}\left(\bar{\delta},\bar{\sigma}\right)=0
\end{align}
subject to the constraint
\begin{align}
\Delta_{\text{\tiny{AB}}}\left(\bar{\delta},\bar{\sigma}\right)
=\Delta_{\text{\tiny{SB}}}\left(\bar{\delta},\bar{\sigma}\right).
\end{align} This occurs at a value of ~$\bar{\delta}{\text{\tiny{AB}}}\approx
2.01$ and~$\bar{\sigma}\approx 1.67$. Optimality is manifest since
\begin{align}
&\Delta_{\text{\tiny{AB}}}\left(\bar{\delta},\bar{\sigma}\right)
=\Delta_{\text{\tiny{SB}}}\left(\bar{\delta},\bar{\sigma}\right)\approx
0.68.
\end{align} The optimal situation is
depicted in Fig.~\ref{fig:GVDJoptimalsigmaP2}(a).

For $P=1$, we express the optimal parameters as
\begin{align}
\delta^\sharp\approx2.01,
\end{align} and
\begin{align}
\sigma^\sharp\approx1.67.
\end{align} For these values, the single query success probability
of the Gaussian model with sharp information cut-off model is
\begin{align}
\text{Pr}_{\checkmark}^{\sharp}\approx 0.68.
\end{align}
This upper bound is approximately 10$\%$ greater than that shown
for the model employing orthogonal states \cite{AHS09}, where
$\text{Pr}_{\checkmark}^{\bot}\approx0.61$ and
$\delta^\bot=\frac{\pi}{2}$.

At first glance, the increase in single-query success probability
of the Gaussian case over the orthogonal case appears somewhat
surprising. The Gaussian wave functions are coherent states and
therefore non-orthogonal~\cite{Pe72}. Intuitively, the orthogonal
states should be optimal especially given that the finite extent
of the momentum wave functions provides a natural fit for encoding
finite infirmation.

Upon closer inspection however, we see that the improvement
results from the ability to `tune' the Gaussian spread,
represented by $\sigma$, to match the encoding length~$P$. No such
`tuning' is possible with the finite states. We depict this in
Fig.~\ref{fig:Comparem}(a) for the constant case with $P=1$ and
optimal $\sigma^\sharp = 1.67$. We see that the encoded momentum
Gaussian wave function is on average narrower than the orthogonal
pulse wave function. Since the momentum and position wave
functions are Fourier transform pairs, narrowing of one results in
broadening of the other.

The subsequent broadening of the encoded Gaussian wave functions
results in a wider optimal measurement window
$\delta^\sharp>\delta^\bot$. This leads to a greater single-query
success probability and is represented by the shaded regions in
Fig.~\ref{fig:Comparem}(b). The larger dark gray region
corresponds to the single-query success probability offered by the
Gaussian wave functions. We thus conclude that the increased
success probability is achieved through the extra degree of
freedom afforded by $\sigma$. For $P=1$, this requires that the
input state be squeezed to $\sigma\approx 1.67$.

\begin{figure}[tbp]
            \begin{center}
            \includegraphics[width=9 cm]{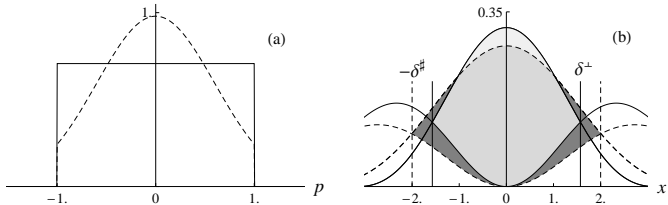}
            \end{center}
            \caption{Comparison between the orthogonal (solid) and
            the Gaussian (dashed) wave functions. The amplitude of the encoded momentum functions for
            the constant case is depicted in (a). The respective algorithm success probabilities are depicted by the shaded regions of
            (b), where light gray correspond to the orthogonal case
            and dark gray corresponds to the Gaussian case. Medium grey corresponds to the region of overlap. Note that the encoded momentum Gaussian is
            on average narrower than the orthogonal pulse, which results in respectively broader position wave functions.} \label{fig:Comparem}
\end{figure}

\section{Conclusions}
\label{sec:conclusions}

We have shown that a simple-harmonic-oscillator quantum computer
solving oracle decision problems performs better using
non-orthogonal Gaussian wave functions as the algorithm input
rather than the orthogonal top-hat wave functions. We have also
shown that the limiting case of the Gaussian model for $\sigma
\rightarrow 0$ and non-zero $P$ corresponds to the model employing
orthogonal states. In both cases, the computational bases are
orthogonal, and encoding takes place in the momentum domain and
information processing and measurement take place in the dual
position domain. Also in both cases, the single-query success
probability is dependent on the maximum separation between the
position wave function encoded with the constant string and the
position wave function encoded with the worst-case balanced
string, which is the antisymmetric balanced string.

In the orthogonal case, $N$-bit strings are uniquely encoded into
the computational basis formed by the top-hat functions, and the
overall-width of the encoded string is set by the encoding length
$P$. In the dual position domain, the encoded string is
represented by a sum of $N$ equi-angularly spaced, equi-length
phasors multiplied by a sinc function. The rate at which the
constant sinc function falls off its peak and the rate that the
antisymmetric balanced sinc function rises from its minimum sets
the size of the optimum position domain measurement window. Thus
the optimum position domain measurement is set by sharpness of the
sinc function, which is dependent on the encoding length only.

In the Gaussian case, $N$-bit strings are uniquely encoded into
the computational basis formed by more complicated
Gaussian-modulated basis states. The overall-width of the encoded
string is again set by the encoding length $P$, but it is also
shaped by the Gaussian spread $\sigma$. In the dual position
domain, the encoded string is represented by a sum of non
equi-angularly spaced and non equi-length phasors multiplied by a
Gaussian function. The rate at which the constant encoded function
falls off its peak and the rate that the antisymmetric balanced
function rises from its minimum is governed by both the Gaussian
spread $\sigma$ and the encoding length $P$. More importantly, the
rate set by the optimal values of $\sigma$ and $P$ is more gradual
than that achievable in the orthogonal case allowing for greater
separation between the two probabilities.

We thus conclude that the Gaussian allows for an improved
trade-off between encoding, processing and measuring of the
information. Encoding takes place in the momentum domain, and the
Gaussian takes better advantage of the space available to encode
the information. Correspondingly, information processing and
measurement take place in the dual position domain. The
Gaussian-encoded position wave function enables a wider
measurement window, which means more of the encoded information is
available for distinguishing between a wave function encoded with
a constant string and a wave function encoded with the worst-case
balanced string.

\section*{Acknowledgements}

We appreciate financial support from the Alberta Ingenuity Fund
(AIF), Alberta Innovates Technology Futures (AITF), Canada's
Natural Sciences and Engineering Research Council (NSERC), the
Canadian Network Centres of Excellence for Mathematics of
Information Technology and Complex Systems (MITACS), and the
Canadian Institute for Advanced Research (CIFAR).

\appendix

\section{Proofs of dominance of symmetric and antisymmetric balanced functions}\label{App:CalcDet}

In this appendix we prove that the symmetric and the antisymmetric
balanced functions maximize the magnitude of
$\left|M^{(N)}_z(x;\sigma,1)\right|$ given by
Eq.~(\ref{Eq:ModulatedPositionWaveFunctionPhi3(x)}) subject to the
balanced condition $\sum_{i=1}^{N} (-1)^{z_i}=0$ in three Lemmas.

%%% LEMMA 1 %%%%%

\begin{lemma}
For the region $|x|\leq\pi$ with $\sigma=0$ and subject to the
balanced condition $\sum_{j=1}^N(-1)^{z_j}=0$,
$\max\left|\phi_{z}^{(N)}(x;\sigma,P)\right|$ occurs for
$z\in\{\text{AB}\}$.
\end{lemma}
\begin{proof}
We prove this Lemma by showing that, in the limiting case where
$\sigma\rightarrow 0$, the encoded position wave function given in
Eq.~(\ref{Eq:ModulatedPositionWaveFunctionPhi3(x)}) becomes the
same as the encoded orthogonal wave function analyzed
in~\cite{AHS09}. In that case, it is proved that the antisymmetric
balanced function dominates all other balanced wave functions in
the region of interest.

We begin by defining the quantity
\begin{align}
A_k(x,\sigma)=\text{erf}\left(\frac{\frac{2Pk}{N}\sigma^2-{\rm i}
x}{\sqrt{2}\sigma}\right)-\text{erf}\left(\frac{\frac{2P(k-1)}{N}\sigma^2-{\rm
i} x}{\sqrt{2}\sigma}\right)\label{Eq:AkDef}
\end{align}
for $k=1,...,N/2$. We use this term here and in later Lemmas. For
ease of understanding the antisymmetric features of this term, we
have elected to change the counting variable in the term
$\vartheta_j$ given in Eq.~(\ref{Eq:VarThetaModulation}) from $j$
to $k$, where $j=k+N/2$. We express the $k^{\rm{th}}$ term of the
encoded position wave function given by
Eq.~(\ref{Eq:ModulatedPositionWaveFunctionPhi3(x)}) in terms of
this quantity as
\begin{align}
\phi_{z}^{(k)}(x;\sigma,P)&=\pm\frac{\eta\,\text{e}^{-\frac{x^2}{2\sigma^2}}}{2\sqrt[4]{\pi}\sqrt{\sigma}}
A_k(x,\sigma)\label{Eq:phikDef}
\end{align} where the $\pm$ represents the effect of the bit
$(-1)^{z_k}$. We represent this quantity as the phasor
\begin{align}
\phi_{z}^{(k)}(x;\sigma,P) &=\pm R_k(x;\sigma,P)\exp\left[{\rm i}
\varphi_k(x;\sigma,P)\right],\label{Eq:phikDef}
\end{align}
to align with the description of the orthogonal case.

The phasor magnitude is expressed
\begin{align}
R_k(x;\sigma,P)=\left|\frac{\eta\,\text{e}^{-\frac{x^2}{2\sigma^2}}}{2\sqrt[4]{\pi}\sqrt{\sigma}}
A_k\right|,\label{Eq:ErfMag}
\end{align}
and the argument is
\begin{align}
\varphi_k(x;\sigma,P)=\arctan\left(\frac{\text{Im}\left[A_k\right]}{\text{Re}\left[A_k\right]}\right),\label{Eq:ErfArg}
\end{align} where we have suppressed the arguments of $A_k$ for the sake of brevity.

The quantities $R_k(x;\sigma)$ and $\varphi_k(x;\sigma)$ are too
opaque to understand limiting behaviour, so we use Taylor series
analysis to gain insight. The Taylor series representation of
angle $\varphi_k(x;\sigma)$ given by Eq.~(\ref{Eq:ErfArg}) is
expressed
\begin{align}
\varphi_k(x;\sigma)&=\frac{(2k-1)x}{N}-\frac{(2k-1)x
\sigma^2}{3N^3}+\frac{2(2k-1)x\sigma^4}{45N^5}\nonumber\\
&+O\left(x \sigma^6\right)+O\left(x^3 \sigma^2\right)
\end{align} where for
$\sigma=0$, we have
\begin{align}
\varphi_k(x;0)&=\frac{(2k-1)x}{N}\label{Eq:Argapprox},
\end{align} which presents an equiangular separation between
subsequent phasors.

Similarly the Taylor series for magnitude given by
Eq.~(\ref{Eq:ErfMag}) is expressed
\begin{align}
R_k(x;\sigma)&=\frac{1}{N \sqrt{\pi }}-\frac{x^2}{6 N^3 \sqrt{\pi
}}+\frac{\left(N^2-12 (k-1) k-4\right)
   \sigma ^2}{6 N^3 \sqrt{\pi }}\nonumber\\
   &+\frac{\left(-5 N^2+60 (k-1)
k+24\right) \sigma ^2 x^2}{180 N^5 \sqrt{\pi
   }}+O\left(x^4\sigma^4\right).
\end{align} For $\sigma=0$, this gives
\begin{align}
R_k(x;0)&=\frac{\sqrt{P}}{N \sqrt{\pi }}-\frac{P^{5/2}x^2}{6 N^3
\sqrt{\pi }}+\cdots+\frac{(-1)^m \sqrt{P}\left(\frac{P
x}{N}\right)^{2 m}}{N
\sqrt{\pi } (2 m+1)!}\nonumber\\
&=\frac{\sin{\left(P x/N\right)}}{\sqrt{P\pi}
x},\label{Eq:MagapproxSinc}
\end{align} where the last step assumes the limit
$m\rightarrow\infty$.

Combining the results of
Eqs.~(\ref{Eq:phikDef}),~(\ref{Eq:Argapprox}), and
~(\ref{Eq:MagapproxSinc}) gives
\begin{align}
\phi_{z}^{(N)}(x;0,P)=\frac{\sin{\left(P x/N\right)}}{\sqrt{P\pi}
x}\sum_{j=1}^{N}(-1)^{z_j}\text{e}^{\rm{i}\left(\frac{N-(2j-1)}{N}\right)P
x},
\end{align}
which is the encoded orthogonal position wave function given in
\cite{AHS09}. The proof given therein suffices to prove Lemma 1.
\end{proof}

%%% LEMMA 2 %%%%%

\begin{lemma}
For $x=0$ and $\sigma>0$ and subject to the balanced constraint
$\sum_{j=1}^N(-1)^{z_j}=0$,
$\max\left|\Xi_z^{(N)}(0;\sigma)\right|$ occurs for $z \in
\{\text{SB}\}$.
\end{lemma}

\begin{proof}

We exploit the structure of the quantity given in
Eq.~(\ref{Eq:AkDef}) with $x=0$ expressed as
\begin{align}
A_k(0,\sigma)=\text{erf}\left(\frac{\frac{2k}{N}\sigma^2}{\sqrt{2}\sigma}\right)
-\text{erf}\left(\frac{\frac{2(k-1)}{N}\sigma^2}{\sqrt{2}\sigma}\right).\label{Eq:AkDef2}
\end{align}
Using the shorthand $A_k=A_k(0,\sigma)$,  we express a set of $N$
terms in the following convenient form
\begin{align}
\{A_{\frac{N}{2}},\ldots,A_{k},\ldots,A_2,A_1,A_1,A_2,\ldots,A_{k},\ldots,A_{\frac{N}{2}}\}
\label{Eq:ErfSumasAk},
\end{align} for $k=1,2,\ldots,\frac{N}{2}$. Note that the $A_k$ are real numbers.

We now show that $A_k>A_{k+1}$. We express the difference between
these terms as
\begin{align} A_k-A_{k+1}&=2\,
   \text{erf}\left(\frac{\sqrt{2} k \sigma
   }{N}\right)
   -\text{erf}\left(\frac{\sqrt{2} (k-1) \sigma }{N}\right)\nonumber\\
   &-\text{erf}\left(\frac{\sqrt{2} (k+1) \sigma }{N}\right).\label{Eq:ErfdiffdeltaK}
\end{align}
Showing that Eq.~(\ref{Eq:ErfdiffdeltaK}) is positive for all $k$
is equivalent to showing that
\begin{align}
2\,\text{erf}(a)-\text{erf}(a+b)-\text{erf}(a-b)>0
\label{Eq:ErfdiffdeltaA}
\end{align}
for $a, b \in \mathbb{R}$ and $a,b>0$.

Over the domain $(0,\infty)$, the error function is strictly
monotonically increasing with strictly monotonically decreasing
slope
$\frac{\text{d}}{\text{d}x}\text{erf}(x)=2e^{-x^2}/\sqrt{\pi}$.
This means that successive increments $\delta x$ result in
decreasing $\delta y=\text{erf}(\delta x)$ increments. This may be
expressed as
\begin{align}
\text{erf}(a)-\text{erf}(a-b)>\text{erf}(a+b)-\text{erf}(a),
\end{align} and thus
\begin{align}
2\text{erf}(a)-\text{erf}(a-b)-\text{erf}(a+b)>0
\label{Eq:ErfdiffdeltaTaylor},
\end{align}
which establishes that $A_k>A_{k+1}$.

The strategy required to maximize the sum of the $N$ terms of the
set~(\ref{Eq:ErfSumasAk}) subject to the balanced constraint is
now clear. Since $A_1>A_2>\cdots >A_{\frac{N}{2}}$, the maximal
term must contain as many of the larger terms as possible. This
maximal sum is thus expressed
\begin{align}
\pm2\left(A_1+A_2+\dots+A_{\frac{N}{4}}-A_{\frac{N}{4}+1}-A_{\frac{N}{4}+2}-\dots-A_{\frac{N}{2}}\right).
\end{align}This expression manifests the symmetric balanced (SB) function definition given
in Eq.~(\ref{eq:string_SB}) thus proving the Lemma.
\end{proof}

%%% LEMMA 3 %%%%%

\begin{lemma}
For $|x|>0$ and $\sigma\geq0$ and subject to the balanced
condition $\sum_{j=1}^N(-1)^{z_j}=0$,
$\max\left|\phi_{z}^{(N)}(x;\sigma)\right|$ occurs for either
$z\in\{\text{SB}\}$ or for $z\in\{\text{AB}\}$.
\end{lemma}

\begin{proof}
We modify the set of elements $A_k$ to include the imaginary
components resulting from $x>0$ as
\begin{align}\label{set:AkandAk*}
\left\{A^*_{\frac{N}{2}}(x,\sigma),\ldots,A^*_{1}(x,\sigma),A_1(x,\sigma),\dots,A_{\frac{N}{2}}(x,\sigma)\right\}.
\end{align} We now exploit the antisymmetric property of this set.  The
fact that $\text{erf}(w^*)=\text{erf}^*(w)$ allows us to use the
notation
\begin{align}A_k(x,\sigma)=\alpha_k+\rm{i}\beta_k,
\end{align} and
\begin{align}A^*_{k}(x,\sigma)=\alpha_k-\rm{i}\beta_k
\end{align}
to capture the overall of effect of the error function having
complex arguments.

The strategy to maximize the sum of the elements in the set given
expression~(\ref{set:AkandAk*}) subject to the balanced constraint
is clear. The sum must be either purely real or purely imaginary.
A complex sum reduces these achievable maximums in two ways. It
causes elements to be subtracted, and it results in a vector sum
rather than a liner sum.

The maximum real sum subject to the balanced constraint is
\begin{align}
&\sum_{k=1}^{k=N/4}A_k(x,\sigma)+A^{*}_k(x,\sigma)-\sum_{k=n/4+1}^{k=N/2}A_k(x,\sigma)+A^{*}_k(x,\sigma)\nonumber\\
&=2\left(\sum_{k=1}^{k=N/4}\alpha_k(x,\sigma)-\sum_{k=N/4+1}^{k=N/2}\alpha_k(x,\sigma)\right),
\end{align}which is achieved for the symmetric balanced function demonstrated in Lemma
2. The maximum imaginary sum subject to the balanced constraint is
\begin{align}
&\sum_{k=1}^{k=N/2}A_k(x,\sigma)-\sum_{k=1}^{k=N/2}A^{*}_k(x,\sigma)\nonumber\\
& =2\rm{i}\sum_{k=1}^{k=N/2}\beta_k(x,\sigma),
\end{align} which is achieved for the antisymmetric balanced
function.

As $x$ increases from zero, the imaginary component of the error
function increases accordingly. For small $x$, the real part still
dominates and the symmetric balanced function is the balanced
function with the greatest magnitude. However, there is a point
where the antisymmetric balanced function takes over the dominate
role. We determine the value of this crossover point, $x_c$, in
terms of $\sigma$ and $P$ in the following.

The $N=4$ case is the simplest case which demonstrates the
crossover. For this case the set is
\begin{align}
\left\{A^*_{2}(x,\sigma),A^*_{1}(x,\sigma),A_1(x,\sigma),A_{2}(x,\sigma)\right\}.
\end{align}
The antisymmetric balanced sum is
\begin{align}
&\left(-A^*_{2}(x,\sigma)-A^*_{1}(x,\sigma)+A_1(x,\sigma)+A_{2}(x,\sigma)\right)\nonumber\\
&=\alpha_1+\alpha_2+\rm{i}(\beta_1+\beta_2)-\alpha_1-\alpha_2+\rm{i}(\beta_1+\beta_2)\nonumber\\
&=\rm{i}2(\beta_1+\beta_2),
\end{align}
and the symmetric balanced sum is
\begin{align}
&\left(-A^*_{2}(x,\sigma)+A^*_{1}(x,\sigma)+A_1(x,\sigma)-A_{2}(x,\sigma)\right)\nonumber\\
&=2(\alpha_1-\alpha_2).
\end{align} The switch over thus occurs when
\begin{align}
(\beta_1+\beta_2)=(\alpha_1-\alpha_2),
\end{align} for which the lowest-order Taylor approximation is
\begin{align}
x_c\approx\frac{P\sigma ^2}{\left(4-P^2\sigma^2\right)}.
\end{align} For $N=8$, this crossover point from symmetric to antisymmetric
dominance is plotted in Fig.~\ref{fig:DominateBalancedFuntions}.
\end{proof}

\bibliographystyle{aiaa-doi}
\bibliography{mymark4}

\end{document}